\documentclass[11pt,runningheads,a4paper]{llncs}

\usepackage{amssymb}
\usepackage{hyperref}
\hypersetup{
  colorlinks   = true,    % Colours links instead of ugly boxes
  urlcolor     = blue,    % Colour for external hyperlinks
  linkcolor    = blue,    % Colour of internal links
  citecolor    = red      % Colour of citations
}

\usepackage{url}
\usepackage[toc,page]{appendix}
\usepackage{graphicx}
\usepackage[rightcaption]{sidecap}%side caption
\usepackage{dirtytalk}%quote
\usepackage[hmarginratio=1:1,top=32mm,columnsep=20pt]{geometry} % Document margins
\usepackage{multicol} % Used for the two-column layout of the document
\usepackage[hang, small,labelfont=bf,up,textfont=it,up]{caption} % Custom captions under/above floats in tables or figures
\usepackage{booktabs} % Horizontal rules in tables
\usepackage{float} % Required for tables and figures in the multi-column environment - they need to be placed in specific locations with the [H] (e.g. \begin{table}[H])
\usepackage{hyperref} % For hyperlinks in the PDF
\usepackage{amssymb,amsmath} %%,amsthm}
\usepackage{lettrine} % The lettrine is the first enlarged letter at the beginning of the text
\usepackage{paralist} % Used for the compactitem environment which makes bullet points with less space between them
\usepackage{verbatim}
\usepackage[ruled,vlined,noline]{algorithm2e}
\usepackage{float}% http://ctan.org/pkg/float
\usepackage{cite}

\usepackage{csquotes}
\graphicspath{ {images/} }

\begin{document}

\mainmatter  % start of an individual contribution

\title{A Space-efficient Parametrized Algorithm for the Hamiltonian Cycle Problem by Dynamic Algebraziation}

\titlerunning{A Space-efficient Parametrized Algorithm for Hamiltonian Cycle by Dynamic Algebraization}

\author{Mahdi Belbasi, Martin F\"urer}
%\thanks{Research supported in part by NSF Grant CCF-1320814.}}
%
\authorrunning{M. Belbasi, M. F\"urer}
%%\authorrunning{Lecture Notes in Computer Science: Authors' Instructions}

\institute{Department of Computer Science and Engineering \\
	Pennsylvania State University \\
	University Park, PA 16802,  USA \\
	\{mahdi,furer\}@cse.psu.edu 
%\url{http://www.cse.psu.edu/~furer}
}


%%\toctitle{Lecture Notes in Computer Science}
%%\tocauthor{Authors' Instructions}
\maketitle

\begin{abstract}
An NP-hard graph problem may be intractable for general graphs but it could be efficiently solvable using dynamic programming for graphs with bounded width (or depth or some other structural parameter). Dynamic programming is a well-known approach used for finding exact solutions for NP-hard graph problems based on tree decompositions. It has been shown that there exist algorithms using linear time in the number of vertices and single exponential time in the width (depth or other parameters) of a given tree decomposition for many  connectivity problems. Employing dynamic programming on a tree decomposition usually uses exponential space. In 2010, Lokshtanov and Nederlof introduced an elegant framework to avoid exponential space by algebraization. Later, F\"urer and Yu modified the framework in a way that even works when the underlying set is dynamic, thus applying it to tree decompositions.\\
In this work, we design space-efficient algorithms to solve the Hamiltonian Cycle and the Traveling Salesman problems, using polynomial space while the time complexity is only slightly increased. This might be inevitable since we are reducing the space usage from an exponential amount (in dynamic programming solution) to polynomial. We give an algorithm to solve Hamiltonian cycle in time $\mathcal{O}((4w)^d\, nM(n\log{n}))$ using $\mathcal{O}(dn\log{n})$ space, where $M(r)$ is the time complexity to multiply two integers, each of which being represented by at most $r$ bits. Then, we solve the more general Traveling Salesman problem in time $\mathcal{O}((4w)^d poly(n))$ using space $\mathcal{O}(\mathcal{W}dn\log{n})$, where $w$ and $d$ are the width and the depth of the given tree decomposition and $\mathcal{W}$ is the sum of weights. Furthermore, this algorithm counts the number of Hamiltonian Cycles.
 \end{abstract}

\section{Introduction}
Dynamic programming (DP) is largely used to avoid recomputing  subproblems. It  may decrease the time complexity, but it uses auxiliary space to store the intermediate values. This auxiliary space may go up to exponential in the size of the input. This means both the running time and the space complexity are exponential for some algorithms solving those NP-complete problems. Space complexity is a crucial aspect of algorithm design, because we typically run out of space before running out of time. To fix this issue, Lokshtanov and Nederlof \cite{lokshtanov2010saving} introduced a framework which works on a static underlying set. The problems they considered were Subset Sum, Knapsack, Traveling Salesman (in time $\mathcal{O}(2^nw)$ using polynomial space),Weighted Steiner Tree, and Weighted Set Cover. They use DFTs, zeta transforms and M\"obius transforms \cite{rota1964foundations, stanley1997enumerative}, taking advantage of the fact that working on zeta (or discrete Fourier) transformed values is significantly easier since the subset convolution operation converts to pointwise multiplication operation. In all their settings, the input is a set or a  graph which means the underlying set for the subproblems is static. F\"urer and Yu \cite{furer2017space} changed this approach modifying a dynamic programming algorithm applied to a tree decomposition (instead of the graph itself). The resulting algorithm uses only polynomial space and the running time does not increase drastically. By working with tree decompositions, they obtain parametrized algorithms which are exponential in the tree-depth and linear in the number of vertices. If the tree decomposition has a bounded width, then the algorithm is both fast and space-efficient. In this setting, the underlying set is not static anymore, because they are working with different bags of nodes. They show that using algebraization helps to save space even if the underlying set is dynamic. They consider perfect matchings in their paper. In recent years, there have been several results in this field where algebraic tools are used to save space when DP algorithms are applied to NP-hard problems. In 2018, Pilipczuk and Wrochna \cite{pilipczuk2018space} applied a similar approach to solve the Minimum Dominating Set problem. Although they have not directly used these algebraic tools but it is a similar approach (in time $\mathcal{O}(3^dpoly(n))$ using $poly(n)$ space).\\
We have to mention that there is no general method to automatically transform dynamic programming solutions to polynomial space solutions while increasing the tree-width parameter to the tree-depth in the running time. \\
One of the interesting NP-hard problems in graph theory is Hamiltonian Cycle. It seems harder than many other graph problems. We are given a graph and we want to find a cycle visiting each vertex exactly once. The naive  deterministic algorithm for the Hamiltonian Cycle problem and the more general the Traveling Salesman problem runs in time $\mathcal{O}(n!)$ using polynomial space. Later, deterministic DP and inclusion-exclusion algorithms for these two problems running in time $\mathcal{O}^*(2^n)$\footnote{$\mathcal{O}^*$ notation hides the polynomial factors of the expression.} using exponential space were given in \cite{karp1982dynamic, kohn1977generating, bellman1962dynamic}. The existence of a deterministic algorithm for Hamiltonian cycle running in time $\mathcal{O}((2-\epsilon)^n)$, for a fixed $\epsilon > 0$ is still an open problem. There are some randomized algorithms which run in time $\mathcal{O}((2-\epsilon)^n)$, for a fixed $\epsilon > 0$ like the one given in \cite{bjorklund2017directed}. Although, there is no improvement in deterministic running time, there are some results on parametrized algorithms.  In 2011, Cygan et al.\ \cite{cygan2011solving} designed a parametrized algorithm for the Hamiltonian Cycle problem, which runs in time $4^{tw}|V|^{\mathcal{O}(1)}$. They also presented a randomized algorithm for planar graphs running in time $\mathcal{O}(2^{6.366\sqrt{n}})$. In 2015, Bodlaender et al.\ \cite{bodlaender2015deterministic} introduced two deterministic single exponential time algorithm for Hamiltonian Cycle: One based on pathwidth running in time $\tilde{\mathcal{O}}(6^{pw}pw^{\mathcal{O}(1)}n^2)$\footnote{$\tilde{\mathcal{O}}$ notation hides the logarithmic factors of the expression.} and the other is based on treewidth running in time $\tilde{\mathcal{O}}(15^{tw}tw^{\mathcal{O}(1)}n^2)$, where $pw$ and $tw$ are the pathwidth and the treewidth respectively. The authors also solve the Traveling Salesman problem in time $\mathcal{O}(n(2+2^\omega)^{pw}pw^{\mathcal{O}(1)})$ if a path decomposition of width $pw$ of $G$ is given, and in time $\mathcal{O}(n(7+2^{\omega+1})^{tw}tw^{\mathcal{O}(1)})$, where $\omega$ denotes the matrix multiplication exponent. One of the best known upper bound for $\omega$ \cite{Williams:2012:MMF:2213977.2214056} is $2.3727$. They do not consider the space complexity of their algorithm and as far as we checked it uses exponential space.\\
Recently, Curticapean et al.\ \cite{curticapean2018tight} showed that there is no positive $\epsilon$ such that the problem of counting the number of Hamiltonian cycles can be solved in $\mathcal{O}^*((6-\epsilon)^{pw})$ time assuming SETH. Here $pw$ is the width of the given path decomposition of the graph. They show this tight lower bound via matrix rank.
\section{Preliminaries}
In this section we review notations that we use later.
%\begin{definition}\textbf{Hamiltonian Path.} A Hamiltonian path in an undirected (directed) graph is an undirected (directed) path that visits all of the vertices of the graph exactly once.
%\end{definition}
%\begin{definition}\textbf{Hamiltonian Cycle.} A Hamiltonian cycle is a Hamiltonian Path which is a cycle (visits the start vertex twice).
%\end{definition}

\subsection{Tree Decomposition}
A \emph{tree decomposition} of a graph $G = (V, E)$, of $G$ is a tree $\mathcal{T} = (V_{\mathcal{T}}, E_{\mathcal{T}})$ such that each node $x$ in $V_{\mathcal{T}}$ is associated with a set $B_x$ (called the bag of $x$) of vertices in $G$, and $\mathcal{T}$ has the following properties:
\begin{itemize}
    \item The union of all bags is equal to $V$. In other words, for each $v\in V$, there exists at least one node $x\in V_{\mathcal{T}}$ with $B_x$ containing $v$.
    
    \item For every edge $\{u, v\} \in E$, there exists a node $x$ such that $u,v\in B_x.$ 
    
    \item For any nodes $x, y \in V_{\mathcal{T}}$, and any node $z\in V_{\mathcal{T}}$ belonging to the path connecting $x$ and $y$ in $\mathcal{T}$, $B_x \cap B_y \subseteq B_z.$
\end{itemize}

The \emph{width of a tree decomposition} is the size of its largest bag minus one. The \textit{treewidth} of a graph $G$ is the minimum width over all tree decompositions of $G$ called $tw(G)$. In the following, we use the letter $k$ for the treewidth. Arnborg et al.\ \cite{arnborg1988problems} showed that constructing a tree decomposition with the minimum treewidth is an NP-hard problem but there are some approximation algorithms for finding  near-optimal tree decompositions \cite{bodlaender1991approximating, bouchitte2004treewidth, bodlaender1996linear}. In 1996, Bodlaender \cite{bodlaender1996linear} introduced a linear time algorithm to find the minimum treewidth if the treewidth is bounded by a constant.\\
To simplify many application algorithms, the notion of a \emph{nice tree decomposition} has been introduced which has the following properties. The tree is rooted and every node in a nice tree decomposition has at most two children. Any node $x$ in a nice tree decomposition $\tau$ is of one of the following types (let $c$ be the only child of $x$ or let $c_1$ and $c_2$ be the two children of $x$):
\begin{itemize}
\item \emph{Leaf} node, a leaf of $\tau$ without any children.
\item \emph{Forget} node (forgetting vertex $v$), where $v \in B_c$ and $B_x = B_c \setminus \{v\}$,
\item \emph{Introduce vertex} node (introducing vertex $v$), where $v \notin B_c$ and $B_x = B_c \cup \{v\}$ ,
%\item \emph{Introduce edge} node (introducing edge $e = \{u, v\})$, where $u, v \in B_x$ and $B_x = B_c$. $e$ is associated with $x$.
\item \emph{Join} node, where $x$ has two children with the same bag as $x$, i.e. $B_x = B_{c_1} = B_{c_2}$.
\end{itemize}
We should mention that in some papers like \cite{furer2017space}, for the sake of simplicity an introduce edge node is defined which is not a part of the standard definition of the nice tree decomposition. Here introduce edge nodes are not needed and we can handle the problem easier without such nodes. In fact, we add edges to the bags as soon as the endpoints are introduced.\\
It has been shown that any given tree decomposition can be converted to a nice tree decomposition with the same treewidth in polynomial time \cite{kneis2009bound}.
%\begin{definition}
%We define $\tau_x = (V_x, E_x)$ to be the subtree rooted at $x \in \tau$ and its edges are edges introduced through the subtree rooted at $x$, where $V_x$ is the union of all bags through the $\tau_x$ and $E_x$ is the set of all edges introduced through $\tau_x$.
%\end{definition}

\begin{definition}
The \emph{depth} of a tree decomposition is the maximum number of distinct vertices in the union of all bags on a path from the root to the leaf. We use $d$ to denote the depth of a given tree decomposition.
\end{definition}
We have to mention that this is different from the depth of a tree. This is the depth of a tree decompisition as defined above.
\begin{definition}
The \emph{tree-depth} of a graph $G$, is the minimum over the depths of  all tree decompositions of $G$.  We use $td$ to denote the tree-depth of a graph.
\end{definition}

After defining the treewidth and the tree-depth, now we have to talk about the relationship between these parameters in a given graph $G$:

\begin{lemma} (see \cite[Corollary 2.5]{nevsetvril2006tree} and \cite{bodlaender1995approximating})

For any connected graph $G$, $td(G) \geq tw(G)+1 \geq \frac{td(G)}{\log_2{|V(G)|}}$.
\end{lemma}

\textbf{Example}: The path $P_n$ with $n$ vertices, has treewidth $1$ and tree-depth of $\log_2(n+1)$.\\
One should note that finding the treewidth and the tree-depth of a given graph $G$ are known to be an NP-hard problem. A question which arises here, is whether there exists a tree decomposition such that its width is the treewidth and its depth is the tree-depth of the original graph? In other words, is it possible to construct a tree decomposition of a given graph which minimizes both the depth and the width? If the answer is yes, how to gain such a tree decomposition. Although we are not focusing on this question here, but it seems an interesting problem to think about.

\subsection{Algebraic tools to save space}
When we use dynamic programming to solve a graph problem on a tree decomposition, it usually uses exponential space. Lokshtanov and Nederlof converted some algorithms using subset convolution or union product into transformed version in order to reduce the space complexity. Later, F\"urer and Yu \cite{furer2017space} also used this approach in a dynamic setting, based on tree decompositions for the Perfect Matching problem. In this work, we introduce algorithms to solve the Hamiltonian cycle and the Traveling Salesman problems. First, let us recall some definitions.\\
Let $\mathcal{R}[2^{\mathcal{V}}]$ be the set of all functions from the power set of the universe $\mathcal{V}$ to the ring $\mathcal{R}$. The operator $\oplus$ is the pointwise addition and the operator $\odot$ is the pointwise multiplication.

\begin{definition}
A \emph{relaxation} of a function $f \in \mathcal{R}[2^{\mathcal{V}}]$ is a sequence of functions $\{f^i : f^i \in \mathcal{R}[2^{\mathcal{V}}], 0 \leq i \leq |\mathcal{V}|\}$, where $\forall \text{ } 0 \leq i \leq |\mathcal{V}|$ and $\forall X \subseteq \mathcal{V}$, $f^i[X]$ is defined as:
\begin{equation}
f^i[X] = 
\begin{cases}
0 & \text{if } i < |X|,\\
f[X] & \text{if } i = |X|,\\
\text{arbitrary value} & \text{if } i > |X|.
\end{cases}
\end{equation}
\end{definition}

\begin{definition} The \emph{zeta transform} of a function $f \in \mathcal{R}[2^{\mathcal{V}}]$ is defined as:
\begin{equation}\label{eq:zeta_trans}
\zeta f[X] = \sum_{Y \subseteq X} f[Y].
\end{equation}
\end{definition}

\begin{definition} The \emph{M\"obius transform} of a function $f \in \mathcal{R}[2^{\mathcal{V}}]$ is defined as:
\begin{equation}\label{eq:mobius_trans}
\mu f[X] = \sum_{Y \subseteq X} (-1)^{|X \setminus Y|} f[Y].
\end{equation}
\end{definition}

\begin{lemma}\label{lem:zeta_mob} The M\"obius transform is the inversion of the zeta transform and vice versa, i.e. 
\begin{equation}\label{eq:zeta_mob}
\mu (\zeta f[X]) = \zeta (\mu f[X]) = f[X].
\end{equation}
\end{lemma}
See \cite{rota1964foundations, stanley1997enumerative} for the proof.\\
Instead of storing exponentially many intermediate results, we store the zeta transformed values. We can assume that instead of working on the original nice tree decomposition, we are working on a mirrored nice tree decomposition to which the zeta transform has been applied. We work on zeta transformed values and finally to recover the original values, we use Equation \ref{eq:zeta_mob}. The zeta transform converts the \say{hard} union product operation ($*_{u}$) to the \say{easier} pointwise multiplication operation ($\odot$) which results in saving space.

\begin{definition}
Given $f, g \in \mathcal{R}[2^{\mathcal{V}}]$ and $X \in 2^{\mathcal{V}}$, the \emph{Subset Convolution} of $f$ and $g$ denoted $(f*_{R}g)$ is defined as:
    \begin{equation}
        (f*_{R}g)[X] = \sum_{X_1 \subseteq X}f(X_1)g(X \setminus X_1).
    \end{equation}
\end{definition}

\begin{definition}
Given $f, g \in \mathcal{R}[2^{\mathcal{V}}]$ and $X \in 2^{\mathcal{V}}$, the \emph{Union Product} of $f$ and $g$ denoted $(f*_{u}g)$ is defined as:
    \begin{equation}
        (f*_{u}g)[X] = \sum_{X_1\cup X_2 = X}f(X_1)g(X_2).
    \end{equation}
\end{definition}

\begin{theorem}\label{thm:union_zeta}
\label{thm:pointwise-multiplication}(\cite{bjorklund2007fourier})
    Applying the zeta transform to a union product operation, results in the pointwise multiplication of zeta transforms of the outputs, i.e., given $f, g \in \mathcal{R}[2^{\mathcal{V}}]$ and $X \in 2^{\mathcal{V}}$, 
    \begin{equation}
        \zeta(f*_{u}g)[X] = (\zeta f) \odot (\zeta g) [X].
    \end{equation}
\end{theorem}

All of the previous works which used either DFT or the zeta transform on a given tree decomposition (such as \cite{furer2017space}, and \cite{lokshtanov2010saving}), have one common central property that the recursion in the join nodes, can be presented by a formula using a union product operation. The union product and the subset convolution are complicated operations in comparison with pointwise multiplication or pointwise addition. That is why taking the zeta transform of a formula having union products makes the computation easier. As noted earlier in Theorem \ref{thm:union_zeta}, taking the zeta transform of such a term, will result in a term having pointwise multiplication which is much easier to handle than the union product. After doing a computation over the zeta transformed values, we can apply the M\"obius transform on the outcome to get the main result (based on Theorem \ref{lem:zeta_mob}). In other words, the direct computation has a mirror image using the zeta transforms of the intermediate values instead of the original ones. While the direct computation keeps track of exponentially many intermediate values, the computation over the zeta transformed values partitions into exponentially many branches, and they can be executed one after another. Later, we show that this approach improves the space complexity using only polynomial space instead of exponential space.

\section{Counting the number of Hamiltonian cycles}
We are given a connected graph $G = (V, E)$ and a nice tree decomposition $\tau$ of $G$ of width $w$. If $H$ is a Hamiltonian cycle, then the induced subgraph of $H$ on $\tau_x$ (called $H[V_x]$, where $V_x$ is the union of all bags in $\tau_x$) is a set of disjoint paths with endpoints in $B_x$ (see Figure \ref{fig:paths}).
\begin{SCfigure}[0.7][h]
\includegraphics[scale=0.5]{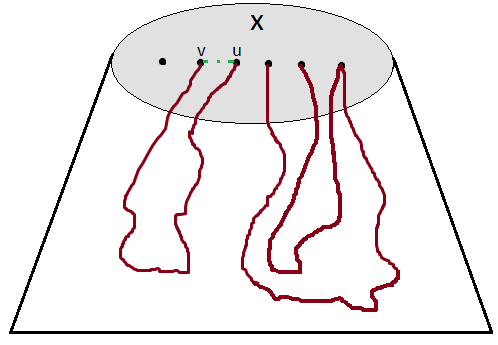}
\caption{$H[V_x]$ is a set of paths with endpoints in $B_x$}
\label{fig:paths}
\end{SCfigure}

\begin{definition}
A \emph{pseudo-edge} is a pair of endpoints of a path of length $\geq 2$ in $H[V_x]$. We use the $\langle , \rangle$ notation for the pseudo-edges. E.g., in Figure \ref{fig:paths}, $p = \langle u, v\rangle$ is a pseudo-edge (it does not imply that there is an edge between $u$ and $v$, it just says that there is a path of length at least two in $H[V_x]$ where $u$ and $v$ are its endpoints). The $\langle , \rangle$ notation is a symmetrical notation since our paths are undirected, i.e., $\langle u, v\rangle = \langle v, u\rangle$. Each path is associated with a pseudo-edge.
\end{definition}

\begin{lemma} \label{lem:pe}
The degree of all vertices in $H[V_x]$ is at most 2.
\end{lemma}
\begin{proof}
$H[V_x]$ is a subgraph of the cycle $H$.
\end{proof}
Let $T_x$ be the vertices contained in the bags associated with nodes in the subtree $\tau_x$ rooted at $x$ which are not in $B_x$. Remember that vertices in pseudo-edges are vertices in $B_x$. Let $X$ be the union of pseudo-edges (in $B_x$). Let $[B_x]^2$ be the set of two-element subsets of $B_x$, and let $X \subseteq [B_x]^2$. Let $S_X$ be the union of vertices involved in $X$. Then, $S_X$ is a subset of $B_x$. For any $X$, define $Y_X$ to be  the union of $S_X$ and $T_x$. Let $F_x$ be the vertices of $G$ which are introduced through $T_x$ and are not present in the bag of the parent of $x$. Define $f_x[X]$ to be the number of sets of disjoint paths (except in their endpoints where they can share a vertex) whose pseudo-edge set is $X$ (remember $S_X$ is the union of vertices involved in $X$) visiting vertices of $F_x$ exactly once (it can also visit vertices which are not in $F_x$ but we require them to visit at least vertices in $F_x$ since they are not present in the proper ancestors of $x$). Computing $f_r[\emptyset]$ gives us the number of possible Hamiltonian cycles in $G$. Now, we show how to compute the values of $f_x$ for all types of nodes.\\
\subsection{Computing $f_x[X]$}\label{sub:f}

In two rounds we will compute $f_x[X]$ efficiently. In the first round we will introduce the recursive formulas for any kind of nodes (when space usage is still exponential) and in the second round we will explain how to compute the zeta transformed values (in the next section, when the space usage drops to polynomial):

\begin{itemize}
\item \textbf{Leaf node}: Assume $x$ is a leaf node.
\begin{equation}
f_x[X] = 
\begin{cases}
1 & \text{if } X = \emptyset,\\
0 & \text{otherwise}.
\end{cases}
\end{equation}
Since $x$ is a leaf node, there is no path through $T_x$, so for all non-empty sets $X$ of pseudo-edges, $f_x[X]$ is zero and for $X = \emptyset$, there is only one set of paths, which is empty.

\item \textbf{Forget node}: Assume $x$ is a forget node (forgetting vertex $v$)  with a child $c$, where $B_x = B_c \setminus \{v\}$. Any pseudo-edge $\langle u, w\rangle \in X \subseteq [B_x]^2$ can define a path starting from $u$, going to $v$ possibly through $T_c$ and then going to $w$ possibly through $T_c$. Here, either or both pieces of the path (from $u$ to $v$, and/or from $v$ to $w$) can consist of single edges.  
\begin{SCfigure}[0.7][h]\label{fig:forget}
\includegraphics[scale=0.5]{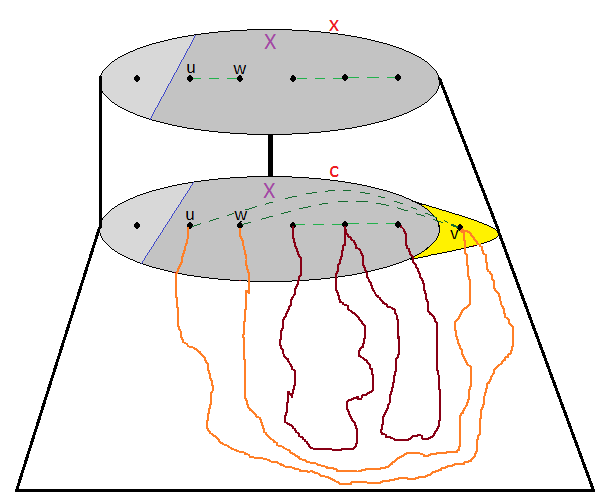}
\caption{Forget node $x$ forgetting vertex $v$ with the child $c$.}

\end{SCfigure}
\begin{equation}\label{eq:forget1}
f_x[X] = \sum_{\langle u, w\rangle \in X} \; \sum_{Q\subseteq \{\langle u, v\rangle, \langle v, w\rangle\}}d_Q \; f_c[X\setminus \{\langle u, w\rangle\}\cup Q],
\end{equation}
where $d_Q = \begin{cases}
1 & \text{ if } \langle u, v \rangle \in Q \cup E \text{ and } \langle v, w \rangle \in Q \cup E \\
0 & \text{ otherwise}.
\end{cases}$

\item \textbf{Introduce vertex node}:  Assume $x$ is an introduce vertex node (introducing vertex $v$)  with a child $c$, where $B_x = B_c \cup \{v\}$. The vertex $v$ cannot be an endpoint of a pseudo-edge because paths have length at least two.
\begin{equation}\label{eq:iv1}
f_x[X] = 
\begin{cases}
f_c[X] & \text{if } v \notin S_X,\\
0 & \text{otherwise}.
\end{cases}
\end{equation}

\item \textbf{Join node}:  Assume $x$ is a join node  with two children $c_1$ and $c_2$, where $B_x = B_{c_1}  = B_{c_2}$. For any given $X \subseteq B_x \times B_x$, $X$ can be partitioned in two sets $X'$ and $X\setminus X'$ and each of them can be the set of pseudo-edges for one of the children. 

\begin{SCfigure}[0.5][h]\label{fig:join}
\includegraphics[scale=0.4]{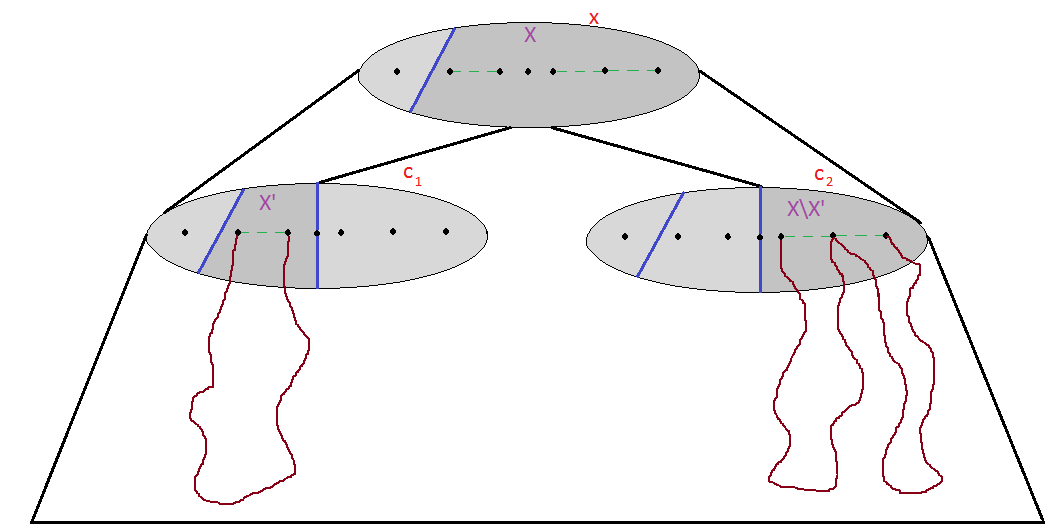}
\caption{Join node $x$ with two children $c_1$ and $c_2$.}
\end{SCfigure}
The number of such paths associated with $X$ through $T_x$ is equal to the sum of the products of the number of corresponding paths associated with $X'$ and $X\setminus X'$ through $T_{c_1}$ and $T_{c_2}$ respectively.
\begin{equation}\label{eq:join1}
f_x[X] = \sum_{X' \subseteq X} f_{c_1}[X'] f_{c_2}[X \setminus X'] = (f_{c_1} *_R f_{c_2})[X].
\end{equation}
Here we get subset convolution and we have to convert it to union product to be able to use zeta transfrom. We do this conversion in the next subsection.

%\item \textbf{Introduce edge node}: To handle the an introduce edge node $x$ (introducing edge $e$ ) with a child $c$ in the original tree decomposition easier, we add an \emph{auxiliary leaf} node as a child of $x$ named $c_2$ such that they have the same bag (i.e. $B_{c_2} = B_x$) and the edge $e$ is introduced in $c_2$. With this setting, introduce edge nodes are a join node with two children. Auxiliary leaf nodes are special leaf nodes where they do not have empty bag. Now we can treat the introduce edge nodes like join nodes.
%\item \textbf{Auxiliary leaf node}: Assume $x$ is an auxiliary leaf node which has only one edge $e = \{v, u\}$. Then
%\begin{equation}
%f_x[X] = 
%\begin{cases}
%1 & \text{if } X = \langle v, u \rangle,\\
%0 & \text{otherwise}.
%\end{cases}
%\end{equation}
\end{itemize}

\subsection{Computing $\zeta f_x[X]$}\label{sub:zetaf}
In this subsection (the second round of computation), first we compute the relaxations of $f_x[X]$ for all kinds of nodes and then we apply the zeta transform to the relaxations. In the following section let $\{f^i_x\}_{0\leq i \leq k}$ be a relaxation of $f_x$.
\begin{itemize}
\item \textbf{Leaf node}: Assume $x$ is a leaf node. Since $f_x[\emptyset] = 1$ and for any $X \neq \emptyset: f_x[X] = 0$, we can choose $f_x^i[X] = f_x[X]$ for all $i$ and $X$. Then
\begin{equation}
(\zeta f_x^{i})[X] = 1, \text{ for all } i \text{ and } X.
\end{equation}
\item \textbf{Forget node}: Assume $x$ is a forget node (forgetting vertex $v$)  with a child $c$, where $B_x = B_c \setminus \{v\}$. Thus,

\begin{equation}
f_x^i[X] = \sum_{\langle u, w\rangle \in X} \sum_{Q\subseteq \{\langle u, v\rangle, \langle v, w\rangle\}}d_Q \, f_c^{i'}[X\setminus \{\langle u, w\rangle\}\cup Q],
\end{equation}
where $i'(Q) = i-1+|Q|$, 
i.e., $i'(Q)$ is the number of pseudo-edges in $X\setminus \{\langle u, w\rangle\}\cup Q$.
Now we apply zeta transform:
\begin{eqnarray} 
\label{eq:forget2}
\lefteqn{(\zeta f_x^i)[X] =   \sum_{Y \subseteq X} f_x^{i}[Y]} \nonumber \\
& = &  \sum_{Y \subseteq X}\: \sum_{\langle u, w\rangle \in Y} \: \sum_{Q\subseteq \{\langle u, v\rangle, \langle v, w\rangle\}}d_Q \, f_c^{i'(Q)}[Y\setminus \{\langle u, w\rangle\}\cup Q]  \nonumber \\
%%new
& = & \sum_{\langle u, w \rangle \in X} \: \sum_{Q\subseteq \{\langle u, v\rangle, \langle v, w\rangle\}} \: \sum_{\{\langle u, w\rangle\} \subseteq Y \subseteq X} d_Q \,  f_c^{i'(Q)} [Y\setminus \{\langle u, w\rangle\}\cup Q]  \nonumber \\
& = & \sum_{\langle u, w \rangle \in X} \: \sum_{Q\subseteq \{\langle u, v\rangle, \langle v, w\rangle\}} \: d_Q \, \sum_{Y \subseteq (X \setminus \{\langle u, w\rangle\})}   f_c^{i'(Q)}[Y \cup Q] 
%% deleted
%& = & \sum_{\langle u, w \rangle \in X} \: \sum_{Q\subseteq \{\langle u, v\rangle, \langle v, w\rangle\}} d_Q \, (\zeta f_c^{i-1+|Q|})[X \setminus \{\langle u, w\rangle\} \cup Q].
\end{eqnarray}
We now express $E_Q =  \sum_{Y \subseteq (X \setminus \{\langle u, w \rangle\})}  f_c^{i'(Q)}[Y \cup Q]$ by $\zeta$-transforms, depending on the size of $Q$. We use the abbreviation $X' = X \setminus \{\langle u, w\rangle\}$ \\

\noindent
If $Q=\emptyset$, then
\begin{eqnarray} 
E_Q = \sum_{Y \subseteq X'}  f_c^{i-1}[Y] 
= \zeta f^{i-1}[X'].
\end{eqnarray}
If $Q = \{\langle u, v \rangle\}$ or $Q =  \{\langle v, w \rangle\}$, then
\begin{eqnarray}
 E_Q = \sum_{Y \subseteq X'}  f_c^{i}[Y \cup Q] 
= \zeta f_c^{i}[X' \cup Q]
	- \zeta f_c^{i}[X']. 
\end{eqnarray}
If $Q = \{\langle u, v\rangle, \langle v, w\rangle\}$, then
\begin{eqnarray}
E_Q & = & \sum_{Y \subseteq X'}  f_c^{i+1}[Y \cup Q] \nonumber \\
	& = & \zeta f_c^{i+1}[X' \cup Q]
	- \zeta f_c^{i+1}[X' \cup \{\langle u, v \rangle\}] \nonumber \\
	&	& \mbox{} - \zeta f_c^{i+1}[X' \cup \{\langle v, w \rangle\}]
	+ \zeta f_c^{i+1}[X'].
\end{eqnarray}

With these sums computed, we can now express $\zeta f_x^i)[X]$ more concisely.
\begin{eqnarray} \label{eq:forget_com}
(\zeta f_x^i)[X] & = & \sum_{\langle u, w \rangle \in X} \: 
	((d_{\emptyset} - d_{\{\langle u, v \rangle\}} - d_{\{\langle v, w \rangle\}} 
	+ d_{\{\langle u, v \rangle\}, \{\langle v, w \rangle\}}) \, \zeta f_c^{i-1}[X']  \nonumber \\
& &	+ d_{\{\langle u, v \rangle\}} \, \zeta f_c^i[X' \cup \{\langle u, v \rangle\}] 
	+ d_{\{\langle v, w \rangle\}} \, \zeta f_c^i[X' \cup \{\langle v, w \rangle\}] \nonumber \\
& &	+ d_{\{\langle u, v\rangle, \langle v, w\rangle\}} \, \zeta f_c^{i+1}[X' \cup \{\langle u, v\rangle, \langle v, w\rangle\}]).
\end{eqnarray}

Note that $d_{\emptyset}=1$ if $\{\langle u, v \rangle, \langle v, w \rangle\} \subseteq E$, $d_{\{\langle u, v \rangle\}}=1$ if $\langle v, w \rangle \in E$, $d_{\{\langle v,w \rangle\}}=1$ if $\langle u,v \rangle \in E$, and $d_{\{\langle u, v \rangle, \langle v, w \rangle\}}$ is always 1, while otherwise $d_Q =0$. This implies that  for every $\langle u, w \rangle$ in the previous equation at least one of the 4 coefficients is 0.
Therefore, in each forget node, we have at most a $3|X|$ fold branching.

\item \textbf{Introduce vertex node}:  Assume $x$ is an introduce vertex node (introducing vertex $v$)  with a child $c$, where $B_x = B_c \cup \{v\}$. Let $X_v$ be the set of pseudo-edges having $v$ as one their endpoints. Therefore,
\begin{equation}
f_x^i[X] = 
\begin{cases}
f_c^i[X] & \text{if } v \notin S_X,\\
0 & \text{otherwise}.
\end{cases}
\end{equation}
\begin{equation}\label{eq:iv2}
(\zeta f_x^i)[X]  = \sum_{Y\subseteq X}f_x^i[Y] = \sum_{Y \subseteq (X \setminus X_v)}f_c^i[Y] = (\zeta f_c^{i})[X \setminus X_v].
\end{equation}

\item \textbf{Join node}:  Assume $x$ is a join node   with two children $c_1$ and $c_2$, where $B_x = B_{c_1}  = B_{c_2}$. To compute $f_x$ on a join node, we can use Equation \ref{eq:join1}. In order to convert the subset convolution operation to pointwise multiplication, first we need to convert the subset convolution to union product and then we are able to use Theorem \ref{thm:union_zeta}. To convert subset convolution to union product we introduce a relaxation of $f_x$. Let $f^i_x[X] = \sum_{j=0}^{i} (f_{c_1}^{j} *_u f_{c_2}^{i-j})[X]$ be a relaxation of $f_x$. 
\begin{equation}\label{eq:join2}
(\zeta f^i_x)[X] = \sum_{j=0}^{i}(\zeta f^j_{c_1})[X] \cdot (\zeta f^{i-j}_{c_2})[X], \text{ for } 0 \leq i \leq w,
\end{equation}
where $w$ is the treewidth of $\tau$.

%\item \textbf{Introduce edge node}: As we mentioned in the previous section, it can be converted to a join node with an auxiliary leaf node as its one of children.

%\item \textbf{Auxiliary leaf node}: Assume $x$ is an auxiliary leaf node which has only one edge $e = \{v, u\}$. Then
%\begin{equation}
%(\zeta f_x)[X] = 
%\begin{cases}
%1 & \text{if } X = \langle v, u \rangle,\\
%0 & \text{otherwise}.
%\end{cases}
%\end{equation}

\end{itemize}

To summerize, we present the following algorithm for the Hamiltonian Cycle problem where a tree decomposition $\tau$ is given.

\begin{algorithm}
    \SetKwInOut{Input}{Input}
    \SetKwInOut{Output}{Output}
    \SetKwInOut{Return}{return}
    \SetKwInOut{EProcedure}{end procedure}
    \SetKwInOut{Procedure}{procedure}

    \Input{A nice tree decompoisition $\tau$ with root $r$}
    \Return{$(\zeta f)(r, \emptyset, 0)$ \quad \quad //$f_r^{0}[\emptyset]$.}
    \textbf{procedure} $(\zeta f)(x, X, i)$: //$(\zeta f)(x, X, i)$ is the representation of $(\zeta f_x^i)[X]$\\
    \qquad \textbf{if} $x$ is a leaf node:\\
    \qquad \qquad \textbf{return} $1$.\\
    \qquad \textbf{if} $x$ is a forget node: //forgetting $v$\\
    \qquad \qquad \textbf{return} the value according to Equation \ref{eq:forget_com}. \\
    \qquad \textbf{if} $x$ is an introduce vertex node: //introducing $v$\\
    \qquad \qquad \textbf{return} $(\zeta f)(c, X \setminus X_v, i)$.\\
    \qquad \textbf{if} $x$ is a join node:\\
    \qquad \qquad \textbf{return} $ \sum_{j=0}^{i}(\zeta f)(c_1, X, j) \cdot (\zeta f)(c_2, X, i-j)$.\\
    \textbf{end procedure}
     
    \caption{Counting the total number of the possible Hamiltonian cycles in a graph given by a nice tree decomposition}\label{ham}
\end{algorithm}

\begin{theorem}\label{thm:t1}
Given a graph $G = (V, E)$ and a tree decomposition $\tau$ of $G$, we can compute the total number of the Hamiltonian Cycles of $G$ in time $\mathcal{O}((4w)^dnM(n\log{n}))$ and in space $\mathcal{O}(wdn\log{n})$ by using Algorithm \ref{ham}, where $w$ and $d$ are the width and the depth of $\tau$ respectively, and $M(n)$ is the time complexity to multiply two numbers which can be encoded by at most $n$ bits. 
\end{theorem}
Proof can be found in appendix.

%\begin{lstlisting}[caption={Counting number of Hamiltonian Cycles},label=list:8-6,captionpos=t,float,abovecaptionskip=-\medskipamount]
%for i:=maxint to 0 do 
%begin 
%   j:=square(root(i));
%end;
%\end{lstlisting}

\section{The traveling Salesman problem}
In the previous section, we showed how to count the total number of possible Hamiltonian cycles of a given graph. In this section, we discuss a harder problem. We know that the Hamiltonian Cycle problem is reducible to the traveling Salesman problem by setting all cost of edges to be 1. We could have just explained how to solve Traveling Salesman problem but we chose to first explain the easier one to help understanding the process. Now, we have most of the notations. First, we recap the formal definition of the traveling Salesman problem.\\

\begin{definition}\textbf{Traveling Salesman.} Given an undirected graph $G = (V, E)$ with weighted edges, where for all $\{u, v\} \in E$, $c_{u,v}$ is the weight (= cost) of $\{u, v\}$ (nonnegative integer). In the traveling Salesman problem we are asked to find a cycle (if there is any) that visits all of the vertices exactly once (it is a Hamiltonian cycle) with minimum cost. 
\end{definition}

%\begin{definition}\label{def:pathcost}
%For a path $P = v_1v_2v_3\dots v_k$ of a graph $G$ with weighted edges, \emph{the cost of $P$} is $c_{_P} = \sum_{i = 1}^{k-1}c_{v_{i}, v_{i+1}}$.  
%\end{definition}

As mentioned, the output should be a minimum cost Hamiltonian cycle. We have the same notations as before. Thus, we are ready to explain our algorithmic framework:\\
The difference between counting the number of Hamiltonian cycles and finding the minimum cost of a Hamiltonian cycle (answer to the Traveling salesman problem), is that we should work with lengths instead of the numbers of solutions. In order to solve this problem, we work with the ring of polynomials $\mathbb{Z}[x]$, where $x$ is a variable. Our algorithm computes the polynomial $P_x(y)[X] = \sum_{i = 0}^{\mathcal{W}} a_i^x[X]y^i$, where $a_i^x[X]$ is the number of solutions in $\tau_x$ associated with $X$ with cost $i$, and $\mathcal{W}$ is the sum of the weights of all edges. The edge lengths have to be nonnegative integers as mentioned above.

\subsection{Computing the Hamiltonian Cycles of All Costs}
To find the answer to the traveling salesman problem, we have find the first nonzero coefficient of $P_r(x)$, where $r$ is the root of the given nice tree decomposition. As we did for the Hamiltonian cycle problem, we show how compute this polynomial recursively for all kinds of node.

\begin{itemize}
\item \textbf{Leaf node}: Assume $x$ is a leaf node. Since we require leaf nodes to have empty bags, then
\begin{equation}
P_x(y)[X] = 
\begin{cases}
1 & \text{if } X = \emptyset,\\
0 & \text{otherwise}.
\end{cases}
\end{equation}

\item \textbf{Forget node}: Assume $x$ is a forget node (forgetting vertex $v$)  with a child $c$, where $B_x = B_c \setminus \{v\}$.
\begin{equation}\label{eq:forgetTS}
P_x(y)[X] = \sum_{\langle u, w\rangle \in X} \sum_{Q\subseteq \{\langle u, v\rangle, \langle v, w\rangle\}}d_Q(P_c(y))[X\setminus \{\langle u, w\rangle\}\cup Q],
\end{equation}
where $d_Q = \begin{cases}
1 & \text{ if } \langle u, v \rangle \in Q \cup E \text{ and } \langle v, w \rangle \in Q \cup E \\
0 & \text{ otherwise}.
\end{cases}$

\item \textbf{Introduce vertex node}:  Assume $x$ is an introduce vertex node (introducing vertex $v$)  with a child $c$, where $B_x = B_c \cup \{v\}$.
\begin{equation}\label{eq:iv1}
P_x(y)[X] =  
\begin{cases}
P_c(y)[X] & \text{if } v \notin S_X,\\
0 & \text{otherwise}.
\end{cases}
\end{equation}

\item \textbf{Join node}:  Assume $x$ is a join node   with two children $c_1$ and $c_2$, where $B_x = B_{c_1}  = B_{c_2}$.
\begin{equation}\label{eq:join1a}
f_x[X] = \sum_{X' \subseteq X}P_{c_1}(y)[X']  P_{c_2}(y)[X \setminus X'].
\end{equation}

We skip the zeta transform part because it is similar to the Hamiltonian Cycle case.
\end{itemize}

\begin{theorem}\label{thm:t2}
Given a graph $G = (V, E)$ and a tree decomposition $\tau$ of $G$, we can solve the Traveling Salesman problem for $G$ in time $\mathcal{O}(4w)^d \cdot poly(n)$ and in space $\mathcal{O}(\mathcal{W}dwn\log{n})$ where $\mathcal{W}$ is the sum of the weights, $w$ and $d$ are the width and the depth of the tree decomposition $\tau$ respectively. 
\end{theorem}
You can find the proof in appendix.

\subsection{Conclusion}
In this work, we solved Hamiltonian Cycle and Traveling Salesman problems with polynomial space complexity where the running time is polynomial in size of the given graph and exponential in tree-depth. Our algorithms for both problems rely on modifying a DP approach such that instead of storing all possible intermediate values, we keep track of zeta transformed values which was first introduced in \cite{lokshtanov2010saving}, and then in \cite{furer2017space} for dynamic underlying sets.

%\item \textbf{Introduce edge node}: To handle the an introduce edge node $x$ (introducing edge $e$ ) with a child $c$ in the original tree decomposition easier, we add an \emph{auxiliary leaf} node as a child of $x$ named $c_2$ such that they have the same bag (i.e. $B_{c_2} = B_x$) and the edge $e$ is introduced in $c_2$. With this setting, introduce edge nodes are a join node with two children. Auxiliary leaf nodes are special leaf nodes where they do not have empty bag. Now we can treat the introduce edge nodes like join nodes.

%\item \textbf{Auxiliary leaf node}: Assume $x$ is an auxiliary leaf node which has only one edge $e = \{v, u\}$. Then
%\begin{equation}
%f_x[X] = 
%\begin{cases}
%1 & \text{if } X = \langle v, u \rangle,\\
%0 & \text{otherwise}.
%\end{cases}
%\end{equation}

\bibliographystyle{amsplain}
\bibliography{HamCyclMBelbasiMFurer}
\newpage

\appendix
\section{Appendix}
Here is the proof of Theorem \ref{thm:t1}.
\begin{proof}
The correctness of the algorithm is shown in Section \ref{sub:f} and Section \ref{sub:zetaf}. Now, we show the running time and the space complexity of our algorithm.\\
\underline{Running time}: The only case that branching happens in where we are handling a forget node. We have at most $w$ branches in forget nodes since the number of pseudo-edges in each bag is bounded by $w+1$ because of Lemma \ref{lem:pe}. Furthermore, based on the formula for forget node, the two unordered pairs of $(u,v)$ and $(v,w)$ can contribute as an edge or a pseudo-edge (four possible cases). On the other hand the number of forget nodes in a path from the root to one leaf is bounded by the $d$ (the depth of tree decomposition $\tau$). Also, the number of vertices is $n$ and we work with numbers of size at most $n!$ (number of paths) which can be represented by at most $n\log{n}$ bits. We handle multiplication of these numbers which happens in $\mathcal{O}(M(n\log{n}))$. All being said, the total running time is $\mathcal{O}((4w)^dnM(n\log{n}))$.\\
\underline{Space complexity}: We keep the results for each strand at a time. The number of nodes in a path from the root to a leaf is bounded by the depth of the tree decomposition ($d$). Along the path, we keep track of bags (of size $w$) and number of disjoint paths (at most $n!$ paths exist which can be shown by at most $n\log{n}$ bits). Therefore, the total space complexity is $\mathcal{O}(wdn\log{n})$.
\end{proof}
And here is the proof of Theorem \ref{thm:t2}.
\begin{proof}
The correctness of the algorithm is shown in Section \ref{sub:f} and Section \ref{sub:zetaf}. Now, we show the running time and the space complexity of our algorithm.\\
\underline{Running time}: The running time analysis is very similar to the analysis done for Hamiltonian Cycle.\\
\underline{Space complexity}: The space complexity analysis is also similar to the analysis done for Hamiltonian Cycle except here we have to keep track of the sum of the wights which at most is $\mathcal{W}$. So there is a factor of $\mathcal{W}$ in the space complexity here.
\end{proof}

%%\printbibliography
\end{document}